\documentclass{article}
\usepackage{graphicx} % Required for inserting images
\usepackage{amsfonts,amsthm}
\usepackage{amsmath,mathtools}   
\usepackage{latexsym}
\usepackage{amssymb}
\usepackage{mathrsfs}
\usepackage{pifont}
\usepackage{tikz,ulem}
\usepackage[all]{xy}
\usepackage{hyperref}
\usepackage{tikz}
\title{Algebraic Leonard trio approach to rational functions: the Hahn case}
\newtheorem{theo}{Theorem}[section]

\newtheorem{rem}{Remark}[section]

\newtheorem{prop}[theo]{Proposition}

\newtheorem{defi}[theo]{Definition}

\def \oV{{\widetilde{V}}}
\def\V{{\mathbb{V}}}
\newcommand{\ox}{\mathrm{x}}

\date{April 2026}

\begin{document}
\author{
Nicolas Cramp\'e\textsuperscript{$1,2$}\footnote{E-mail: crampe1977@gmail.com}~,
Quentin Labriet\textsuperscript{$3$}\footnote{E-mail: quentin.labriet@umontreal.ca}~,
Lucia Morey\textsuperscript{$3$}\footnote{E-mail: lucia.morey@umontreal.ca}~,\\
Luc Vinet\textsuperscript{$3$}\footnote{E-mail: luc.vinet@umontreal.ca}~,
\vspace{0.5cm}\\
\textsuperscript{$1$}
\small Institut Denis-Poisson CNRS/UMR 7013 - Universit\'e de Tours \\
\small Parc de Grammont, 37200 Tours, France.\vspace{0.2cm}\\
\textsuperscript{$2$}
\small 
Laboratoire d’Annecy de Physique Th\'eorique, 9 Chemin de Bellevue \\
\small  BP 110 - Annecy-le-Vieux -
F-74941 Annecy Cedex - France.\vspace{0.2cm} \\
 \textsuperscript{$3$}
\small Centre de Recherches Math\'ematiques, Universit\'e de Montr\'eal, P.O. Box 6128, \\
\small Centre-ville Station, Montr\'eal (Qu\'ebec), H3C 3J7, Canada.}
\maketitle

\begin{center}
\begin{minipage}{12cm}
\begin{center}
{\bf Abstract}\\
\end{center}  
The finite families of Hahn polynomials and associated biorthogonal rational functions are interpreted algebraically in the framework of Leonard trios. We introduce the trio Hahn algebra  and prove that it is isomorphic to the meta Hahn algebra, thereby clarifying the structural connection between Leonard trios and meta algebras. Finite dimensional realizations in terms of difference operators are constructed, and the functions of interest arise as overlaps between eigensolutions of ordinary  eigenvalue problems. Their bispectral and biorthogonality properties follow naturally from the algebraic framework. 
\end{minipage}
\end{center}

\medskip

\section{Introduction}

Orthogonal polynomial families in the ($q$-)Askey scheme \cite{Koekoek} exhibit a characteristic bispectral structure: they satisfy both a three-term recurrence relation and a differential or difference equation, each of which can be interpreted as an eigenvalue equation for suitable linear operators acting either on the variable or on the polynomial degree. These operators generate the Askey--Wilson algebra, or one of its various degenerations and specializations \cite{zhedanov1991hidden}. Within this setting, the orthogonal polynomials arise as transition coefficients between eigenbases associated with two generators of the algebra in an irreducible representation. In the finite-dimensional case, the corresponding operators form a Leonard pair \cite{TERWILLIGER2001}, a structure known to classify the terminating families in the ($q$-)Askey scheme.

An algebraic approach to biorthogonal rational functions has been developed for the Hahn type \cite{VINET2021124863,Tsujimoto_Vinet_Zhedanov_2025} and the $q$-Hahn type \cite{bussiere2022bispectrality,bernard2024meta}, and was subsequently extended to the Racah type \cite{MetaRacah} and Wilson rational functions \cite{CTVZ25}. These works led to the introduction of meta algebras. The finite-dimensional representation theory of the meta ($q$-)Hahn and Racah algebras was established in \cite{Tsujimoto_Vinet_Zhedanov_2025,bernard2024meta,MetaRacah}, where purely algebraic derivations of the defining relations for the corresponding rational functions were obtained.

More recently, the notion of Leonard trios was introduced \cite{LT2026} and  their connection with biorthogonal rational functions was established. This perspective provides a new algebraic setting for the study of these functions and opens several promising directions for future research. It also offers an alternative viewpoint to the meta algebra program while retaining a number of common structural features. A classification of Leonard trios could potentially lead to a complete scheme for bispectral rational functions analogous to the ($q$-)Askey scheme. Such a classification would likely require determining the algebraic relations satisfied by the operators $V$, $\oV$, and $Z$ forming a Leonard trio (see Definition \ref{def:LT} below). In this direction, one naturally expects a generalization of the Askey--Wilson relations, which characterize the interplay between the two elements of a Leonard pair $(X,Y)$. The present work constitutes a first step toward the development of such a theory.

The paper is organized as follows. In Section~\ref{sec:algebras}, we introduce the algebraic structures associated with Hahn rational functions. In particular, we show how the meta Hahn algebra gives rise to the trio Hahn algebra, namely the algebra associated with the corresponding Leonard trio, and prove that these two algebras are isomorphic. In Section~\ref{sec:realization}, we provide a realization of these algebras in terms of difference operators and compute the corresponding (generalized) eigenvectors. In Section~\ref{sec:hahn-polynomials}, we explain how Hahn polynomials and their properties arise naturally within this framework. Finally, in Section~\ref{sec:rational-functions}, we study the associated biorthogonal rational functions and analyze their main properties.

\section{Algebraic structures for Hahn rational functions\label{sec:algebras}}

Throughout this paper, we use the notation  $[A,B]=AB-BA$ and $\{A,B\}=AB+BA$  for the commutator and the anti-commutator, respectively.

\subsection{Meta Hahn algebra}
The meta Hahn algebra was introduced in \cite{VINET2021124863,Tsujimoto_Vinet_Zhedanov_2025}, and finite families of biorthogonal rational functions and Hahn type orthogonal polynomials were interpreted algebraically by considering its finite-dimensional representations. We recall its definition for the reader’s convenience (see \cite[Def. 3.1]{Tsujimoto_Vinet_Zhedanov_2025}). 
\begin{defi}
The meta Hahn algebra $m\mathfrak{H}$ is the associative algebra with unit $I$, generators $X$, $V$, $Z$ and $\eta$, $\xi$ central elements obeying the defining relations
\begin{align}
&[Z,X]=Z^2+Z\label{eq:rel-meta1}\,,\\
&[V,Z]=2X+\eta \label{eq:rel-meta2}\,,\\
&[X,V]=\{V,Z\}+V+\xi \label{eq:rel-meta3}\,.
\end{align}
\end{defi}
The meta Hahn algebra possesses a Casimir element given by
\begin{align}\label{eq:Cas}
   Q=ZVZ+VZ+X^2-(I-\eta)X+\xi Z\,,
\end{align}
which commutes with the generators $ X$, $Z$ and $ V$. 
Next we discuss important connections between the meta Hahn algebra and some related algebras.
\paragraph{Jacobi algebra.} Let remark that $X$ can be expressed in terms of $V$ and $Z$ using relation \eqref{eq:rel-meta2}. Substituting this expression, one obtains an equivalent set of defining relations for the meta Hahn algebra involving only $V$ and $Z$:
\begin{align}
&[Z,[V,Z]]=2Z^2+2Z\,,\\
&[[V,Z],V]=2\{V,Z\}+2V+2\xi\,.
\end{align}
This algebra is named after Jacobi, as it is the bispectral algebra associated with Jacobi polynomials. Indeed, one can verify that the following linear map
\begin{align}
    Z &\mapsto (\ox-1)/2\,,\\
    V &\mapsto (1-\ox^2)\partial_x^2+(b-a-(a+b+2)\ox)\partial_x\,,
\end{align}
where $\ox$ is the operator multiplication by $x$ and $\partial_x$ is the partial derivative, 
is an algebra homomorphism (with $2\xi=(b^2-a^2)I$).

\paragraph{Hahn algebra.} The Hahn algebra embeds in the meta Hahn algebra. Setting 
\begin{equation}\label{eq:defK1K2}
    K_1 = X + \rho Z+\frac{1}{2}(\eta+\rho I)\,,\quad K_2 = V,
\end{equation}
with $\rho\in\mathbb{R}$, and using the relations \eqref{eq:rel-meta1}--\eqref{eq:rel-meta3} one see that $K_1$ and $K_2$ satisfy
\begin{align}
&[K_1, [K_2, K_1]]= 2K_1^2-K_2 -2Q-\frac{1}{2}(\eta^2+\rho^2I)+\eta-\xi\,,\\
&[[K_2, K_1], K_2]]= 2\{K_1, K_2\}  + 2\rho \xi\,. 
\end{align}
This algebra corresponds to the one realized by the bispectral operators of the Hahn polynomials. Let us notice that, in some previous papers, the Hahn algebra corresponds to a more general  algebra with more parameters which also encompasses the algebra defined in the previous paragraph.\\

As mentioned in the introduction, the meta algebra has been defined to provide an algebraic framework to the study of the bispectral properties of some rational functions. Indeed, for finite representations of the meta Hahn algebra just introduced, these functions appear as overlap coefficients between the bases $\{e_k\}$ and $\{d_\ell\}$ defined by the following (generalized) eigenvalue problems:
\begin{align}
    V e_k =\lambda_k e_k\,,\quad X d_\ell =\mu_\ell Z d_\ell\,.\label{eq:GEVP}
\end{align}
In order for the vectors $d_\ell$ defined by the second equation to define a basis, we assume that $Z$ is invertible, \textit{i.e.} that there exists a $Z^{-1}$ such that
\begin{equation}
ZZ^{-1}=Z^{-1}Z=I\,.
\end{equation}
We suppose that this holds throughout the rest of this paper.
We can then define $\widetilde{V}$ as 
\begin{align}
    \widetilde{V}=XZ^{-1},
\end{align}
which allows us to rewrite the generalized eigenvalue problem \eqref{eq:GEVP} as a standard eigenvalue problem
\begin{align}
\widetilde{V}\left(Zd_\ell\right)=\mu_\ell\left(Zd_\ell\right)\,.
\end{align}
It is then natural to focus on the three operators $V$, $\widetilde{V}$ and $Z$.
We shall show that these three operators, in finite representations, form a Leonard trio, according in \cite{LT2026} that is recalled in Subsection \ref{ssec:LT}. We show in the next subsection that there exists an algebra, called trio algebra, associated to these three generators.

\subsection{Trio Hahn algebra}
When $Z$ is invertible, the relations \eqref{eq:rel-meta1}–\eqref{eq:rel-meta3} of the meta Hahn algebra $m\mathfrak{H}$ yield additional relations. In particular, \eqref{eq:rel-meta1} implies
\begin{align}
&[Z,\widetilde{V}]=[Z,\widetilde{V}Z]Z^{-1}=[Z,X]Z^{-1}=(Z^2+Z)Z^{-1}=Z+I\label{eq:rel-mH-tH-4}\,,\\
&[\oV,Z^{-1}]=Z^{-1}[Z,\oV]Z^{-1}=Z^{-1}(Z+I)Z^{-1}=Z^{-1}+Z^{-2}\label{eq:rel-mH-tH-5}\,.
\end{align}
Equation \eqref{eq:rel-meta2} implies
\begin{align}
&[V,Z]=2\widetilde{V}Z+\eta=\{\widetilde{V},Z\}+[\widetilde{V},Z]+\eta=\{\widetilde{V},Z\}-Z+\eta-I\label{eq:rel-mH-tH-3}\,,\\
&[V,Z^{-1}]=Z^{-1}[Z,V]Z^{-1}=-\{\widetilde{V},Z^{-1}\}-(\eta-I)Z^{-2}+Z^{-1}\,.\label{eq:rel-mH-tH-6}
\end{align}
Equation \eqref{eq:Cas} implies
\begin{equation}
QZ^{-1}=ZV+V+\xi+\oV Z\oV+(\eta-I)\oV\,,
\end{equation}
and then using \eqref{eq:rel-meta3}
\begin{align}
[V,\widetilde{V}]&=[V,X]Z^{-1}+X[V,Z^{-1}]\nonumber\\
&=-(\{V,Z\}+V+\xi)Z^{-1}+\widetilde{V}Z(-\{\widetilde{V},Z^{-1}\}+(I-\eta)Z^{-2}+Z^{-1})\nonumber\\
&=-\widetilde{V}^2-V+\widetilde{V}-(ZV+V+\xi+\widetilde{V}Z\widetilde{V}+(\eta-I)\widetilde{V})Z^{-1}\nonumber\\
&=-\widetilde{V}^2-V+\widetilde{V}-Q Z^{-2}\,.\label{eq:rel-mH-tH-2}
\end{align}
These relations suggest to define the following algebra.
\begin{defi}
The trio Hahn $t\mathfrak{H}$ algebra is the associative algebra with unit $I$ and generators $\mathcal{V}, \mathcal{\oV},\mathcal{Z}, \mathcal{Z}^{-1}$ obeying the relations 
\begin{align}
&\mathcal{Z}^{-1}\mathcal{Z}=\mathcal{Z}\mathcal{Z}^{-1}=I\label{eq:rel-algebra-trio1}\,,\\
&[\mathcal{V},\mathcal{\widetilde{V}}]=-\mathcal{\widetilde{V}}^2-\mathcal{V}+\mathcal{\widetilde{V}}+\zeta \mathcal{Z}^{-2}\,,\label{eq:rel-algebra-trio2}\\
&[\mathcal{V},\mathcal{Z}]=\{\mathcal{\widetilde{V}},\mathcal{Z}\}-\mathcal{Z}+\eta-I\label{eq:rel-algebra-trio6}\,,\\
&[\mathcal{Z},\mathcal{\widetilde{V}}]=\mathcal{Z}+I\label{eq:rel-algebra-trio3}\,,
\end{align}
where $\zeta$ and $\eta$ are central elements.
\end{defi}
The trio Hahn algebra has a Casimir element given by
\begin{equation}
\label{eq:casimir-trios}
\mathcal{C}=\mathcal{Z}\mathcal{V}+\mathcal{V}+\zeta\mathcal{Z}^{-1}  +\mathcal{\widetilde{V}Z\widetilde{V}}+(\eta-I)\mathcal{\widetilde{V}}\,,
\end{equation}
which commutes with the generators $\mathcal{V}, \mathcal{\oV}, \mathcal{Z}$ and $\mathcal{Z}^{-1}.$ The relations of  $t\mathfrak{H}$ imply
\begin{align}
&[\mathcal{Z},\mathcal{\widetilde{V}Z}]=[\mathcal{Z},\mathcal{\widetilde{V}}]\mathcal{Z}=(\mathcal{Z}+I)\mathcal{Z}=\mathcal{Z}^2+\mathcal{Z}\,,\label{eq:rel-trios-meta-1}\\
 &[\mathcal{V},\mathcal{Z}]=\{\mathcal{\widetilde{V}},\mathcal{Z}\}-\mathcal{Z}+\eta-I=2\mathcal{\oV}\mathcal{Z}+\eta\,,\label{eq:rel-trios-meta-2}\\
&[\mathcal{V},\mathcal{\widetilde{V}Z}]=-\mathcal{VZ}+\zeta \mathcal{Z}^{-1}+\mathcal{\widetilde{V}Z\widetilde{V}}+(\eta-I)\mathcal{\widetilde{V}}=-\{\mathcal{V},\mathcal{Z}\}-\mathcal{V}+\mathcal{C}\label{eq:rel-trios-meta-3}\,.
\end{align}
\begin{rem}
Relation \eqref{eq:rel-algebra-trio3} is equivalent to 
$
[\mathcal{\widetilde{V}},\mathcal{Z}^{-1}]=\mathcal{Z}^{-2}+\mathcal{Z}^{-1}\,,
$
and relation \eqref{eq:rel-algebra-trio6} is equivalent to 
$
    [\mathcal{V},\mathcal{Z}^{-1}]=-\{\mathcal{\oV},\mathcal{Z}^{-1}\}-(\eta-I)\mathcal{Z}^{-2}+\mathcal{Z}^{-1}\,.
$
\end{rem}
\begin{prop}
    The meta Hahn algebra $m\mathfrak{H}^\star$ (with the star $^\star$ indicating that $m\mathfrak{H}$ is extended with the inverse operator $Z^{-1}$ satisfying $ZZ^{-1}=Z^{-1}Z=I$) and the trio Hahn algebra $t\mathfrak{H}$ are isomorphic. 
\end{prop}
\begin{proof}
Let $\varphi:m\mathfrak{H}^\star\to t\mathfrak{H}$ be the linear map defined by 
\begin{align}
\label{eq:map-meta-trio}
    &V\mapsto \mathcal{V}\,,\quad Z\mapsto \mathcal{Z}\,,\quad X\mapsto \mathcal{\oV}\mathcal{Z}\,,\quad Q\mapsto -\zeta\,,\quad \xi\mapsto -\mathcal{C}\,,\quad \eta \mapsto \eta\,.
\end{align}
Equations \eqref{eq:rel-mH-tH-4}--\eqref{eq:rel-mH-tH-2} imply $\varphi[A,B]=[\varphi(A),\varphi(B)]$ for all $A,B\in m\mathfrak{H}^\star$ and we conclude that $\varphi$ is an algebra morphism.
Let $\psi:t\mathfrak{H}\to m\mathfrak{H}^\star$ be the linear  map defined by 
\begin{align}
&\mathcal{V}\mapsto V\,,\quad \mathcal{Z}\mapsto Z\,,\quad \mathcal{\oV}\mapsto XZ^{-1}\,,\quad \zeta\mapsto -Q\,,\quad \mathcal{C}\mapsto -\xi\,,\quad \eta \mapsto \eta\,.
\end{align}
Equations \eqref{eq:rel-trios-meta-1}--\eqref{eq:rel-trios-meta-3} imply that $\psi$ is an algebra morphism. Moreover $\varphi\,\circ\, \psi=\psi\,\circ\,\varphi=1$ and we conclude that $m\mathfrak{H}^\star$ and $t\mathfrak{H}$ are isomorphic. 
\end{proof}

\section{Difference realization of the algebras}
\label{sec:realization}

In this section, we consider the realization of the different algebras introduced previously as  operators acting on $\mathbb{C}_N[x]$, the space of polynomials  in the variable $x$ of degree at most $N$.

Let $\ox$ be the operator multiplication by $x$, and  $T^\pm$ be the shift operators defined by 
\begin{equation}
\ox f(x)=xf(x)\,,\qquad T^\pm f(x)=f(x\pm1)\,,
\end{equation}
where $f\in \mathbb{C}_N[x]$. These operators satisfy the relations
\begin{align}\label{eq:commxT}
    T^{\pm} \ox=(\ox \pm 1)T^{\pm}\,.
\end{align}

\paragraph{Realization of the meta Hahn algebra.}
Using relations \eqref{eq:commxT}, the following operators 
\begin{align}
&V=(\ox+a)(\ox+1-a-N)T^+-\ox(\ox+1)I\,,\\
&X=\ox T^--(\ox+c)I\,,\\
&Z=-T^-\,,
\end{align}
are seen to satisfy the defining relations of the meta Hahn algebra with $\eta=(2c+N)I$, $\xi=(1-a)(a+N)I$. 
The Casimir element in this realization takes the form
\begin{align}
    Q=(a-c)(a+c + N - 1)I\,.
\end{align}
By abuse of notation, we use the same letters for the abstract elements of the algebra and their realization as operators acting on $\mathbb{C}_N[x]$. 
Let us emphasize that, in this realization, the operator $Z$ is invertible, $Z^{-1}=-T^+$, and, by consequence, the previous operators provide also a realization of the algebra $m\mathfrak{H}^\star$.

\paragraph{Realization of the trio Hahn algebra. \label{par:realization-trio}}
Due to the isomorphism given in \eqref{eq:map-meta-trio}, we obtain the following realization of the trio Hahn algebra, starting from that associated with the meta Hahn algebra:
\begin{align}
&V=(\ox+a)(\ox+1-a-N)T^+-\ox(\ox+1)I\,,\\
&\oV=(\ox+c)T^+-\ox I\,,\\
&Z=-T^-\,,\\
&Z^{-1}=-T^+\,.
\end{align}
with $\zeta=(c-a)(a+c+N-1)I$ and $\eta=(2c+N)I$.
The Casimir element of the trio Hahn algebra reduces to a scalar:
\begin{align}
    C=(a-1)(a+N)I\,.
\end{align}
\begin{rem}
A more general realization of the meta and trio Hahn algebra is given as follows:
\begin{align}
&V=(\ox+a)(\ox+b)T^+-\frac14(N+a+b+2\ox+1)(N+a+b+2\ox-1)I\,,\\
&\widetilde V=(\ox+c)T^+-\frac12(2\ox+N+a+b-1)I\,, \\
&X=\frac12\left(2\ox+N +a +b -1\right) T^--\left(c +\ox \right)I\,,\\
&Z=-T^-\,,\\
&Z^{-1}=-T^+\,,
\end{align}
where $\eta=(2c-a-b+1)I,$ $\xi=\frac{1}{4}\left(a-b +1 +N \right) \left(b-a +1+N \right)I$ and $\zeta=(b-c)(a-c)I$.
The Casimir elements of the trio and meta Hahn are the following scalars in this realization: 
\begin{align}
&Q=(a-c)(c-b)I\,,\\
&C=\frac{1}{4}\left(a-b +1 +N \right) \left(a-b -1-N \right)I\,.
\end{align}
We note that, for $b=1-a-N$, the operators reduce to those given previously.
\end{rem}

\paragraph{Eigenvalue problems.} In the realizations given above, we can find the polynomials in $\mathbb{C}_N[x]$ which diagonalize the different operators.

The basis $(\mathbf{a}_n)_{n=0}^N$ of $\mathbb{C}_N[x]$ defined by
\begin{align}
  \mathbf{a}_n(x)=(x+a)_n(x+1-a-N)_{N-n}\,,\label{eq:defn-en}  
\end{align}
diagonalizes the operator $V$, for $n=0,\ldots,N$:
\begin{align}
&V\mathbf{a}_n=(1-a-n)(a+n)\mathbf{a}_n\,.\label{eq:eigVa}
\end{align}
The proof is based on the action of $T^+$ on polynomials and using properties of the Pochhammer symbols. Indeed, the action of $V$ reads:
\begin{align}
  V\mathbf{a}_n(x)=&(x+a)(x+1-a-N) (x+1+a)_n(x+2-a-N)_{N-n}\nonumber\\
  &-x(x+1)(x+a)_n(x+1-a-N)_{N-n}\,.
\end{align}
Then, remarking that $(x+a) (x+1+a)_n=(x+a+n)(x+a)_n$ and $(x+1-a-N)(x+2-a-N)_{N-n}=(x+1-a-n)(x+1-a-N)_{N-n}$, we prove relation \eqref{eq:eigVa}.

Similarly, we can show that the operator $\oV$ is diagonalized by the basis $(\mathbf{b}_n)_{n=0}^N$ of $\mathbb{C}_N[x]$ as follows
\begin{align}
&\oV\mathbf{b}_n=(n+c)\mathbf{b}_n\,,\label{eq:eigoVb}
    \end{align}
 with the following eigenvectors:
 \begin{align}
     \mathbf{b}_n(x)=(x+c)_n\,.\label{eq:defn-bn}
 \end{align}
Finally, the eigenvalue problem for the operator $K_1$ defined in \eqref{eq:defK1K2}
\begin{align}
\label{eq:DefK1realization}
 K_1=X + \rho Z+\frac{1}{2}(\eta+\rho I)  =(\ox-\rho)(T^--1)+\frac{N-\rho}{2}\,, 
\end{align} 
reads
\begin{align}
& K_1\mathbf{c_n}=\left((N-\rho)/2-n\right) \mathbf{c}_n\,,
\end{align}
where the eigenvectors $\mathbf{c}_n$ are defined by
\begin{align}
    \mathbf{c}_n(x)=(\rho-x)_n\,.\label{eq:defn-fn}
\end{align}
%Notice that each of the families $\mathbf{a}_n, \mathbf{b}_n$ and $\mathbf{f}_n$ is a basis of $\sV_N$. 
\begin{rem}
   Recalling that $X=\widetilde VZ$, we obtain the following generalized eigenvalue problem from the eigenvalue problem \eqref{eq:eigoVb}:
\begin{align}
&X\mathbf{d}_n=(n-c)Z\mathbf{d}_n\,,
\end{align}
where the generalized eigenvectors are given by
\begin{align}
&\mathbf{d}_n(x)=Z^{-1}\mathbf{b}_n(x)=-(x+c+1)_n\,. \label{eq:defn-dn}
\end{align}
\end{rem}

\section{Hahn polynomials\label{sec:hahn-polynomials}}

This section recalls how the algebraic setup presented above allows us to derive properties of the Hahn polynomials.
Indeed, as mentioned previously, these polynomials arise in the  overlap coefficients between the eigenbases $\mathbf{a}_k$ of the operator $V$ and the eigenbases $\mathbf{c}_\ell$ of the operator $K_1$, defined in \eqref{eq:defn-en} and \eqref{eq:defn-fn}, respectively. This interpretation of these polynomials allows us to recover their properties algebraically.

\begin{prop}
\label{prop:overlap-formulas-pols}The following connection formulas hold
    \begin{align}
&\mathbf{a}_k=(a+\rho)_k(1-a-N+\rho)_{N-k}\sum_{\ell=0}^N\frac{(-N)_\ell}{\ell!(1-a-N+\rho)_\ell} Q_k(\ell) \mathbf{c}_\ell\,,\label{eq:overlap-ab-Hahn}\\
&\mathbf{c}_\ell=(-1)^N(a+\rho)_\ell\sum_{k=0}^N\frac{(2k+2a-1)(-N)_k}{k!(k+2a-1)_{N+1}}Q_k(\ell)\mathbf{a}_k\,,\label{eq:overlap-ba-Hahn}
\end{align}
where the Hahn polynomials is defined by
\begin{equation}
    Q_k(\ell)={}_{3}F_2 \left({{-k,-\ell,k+2a-1 }\atop
{ a+\rho,-N }}\;\Bigg\vert \; 1\right)\,.
\end{equation}
\end{prop}
\proof Let introduce the following basis of $\mathbb{C}_N[x]$:
\begin{align}
\label{eq:basis-sn}
    \mathbf{s}_n(x)=(x+a)_n\,.
\end{align}
We can show that the operators $K_1$ and $V$ act bidiagonally in this basis:
\begin{align}
    K_1\mathbf{s}_n&=((N-\rho)/2-n)\mathbf{s}_n+n(n-1+a+\rho)\mathbf{s}_{n-1}\,,\\
    V\mathbf{s}_n &=(1-a-n)(n+a) \mathbf{s}_n+(n-N) \mathbf{s}_{n+1}\,.\label{eq:actVs}
\end{align}
The diagonalization of these operators is heance very easy
\begin{align}
     \mathbf{a}_k&=\sum_{i=k}^N (-1)^{i+N}\binom{N-k}{i-k}(k+2a+i)_{N-i}\ \mathbf{s}_i\,,\\
     \mathbf{s}_i&=\sum_{\ell=0}^i (-1)^\ell\binom{i}{\ell}(\ell+a+\rho)_{i-\ell}\ \mathbf{c}_\ell\,.
\end{align}
Combining these results, we prove \eqref{eq:overlap-ab-Hahn}. The proof of relation \eqref{eq:overlap-ba-Hahn} follows the same steps.
\endproof
The basis $\mathbf{s}_n$ introduced in the proof corresponds to the basis denoted by $|N-n\rangle$ in \cite{Tsujimoto_Vinet_Zhedanov_2025} (with their parameters modified as follows: $a_n\to N-n,\ \alpha\to N+c,\ \beta\to a+N-1$), it is also the split basis introduced in \cite{TERWILLIGER2001} in the context of Leonard pair.

Let us mention that the previous proposition can also be demonstrated taking limits of the formulas proven in \cite{Rosengren} for the Racah polynomials.

\paragraph{Orthogonality.}
Equations \eqref{eq:overlap-ab-Hahn} and \eqref{eq:overlap-ba-Hahn} allow us to recover the orthogonality relation of the Hahn polynomials:
\begin{align}
&\sum_{\ell=0}^N\binom{a+\rho-1+\ell}{\ell}\binom{a+N-1-\rho-\ell}{N-\ell} Q_m(\ell)Q_n(\ell)\nonumber\\
&=\frac{(-1)^nn!(n+2a-1)_{N+1}(a-\rho)_n }{(2n+2a-1)(a+\rho)_n(-N)_nN!}\delta_{m,n}\,.
\end{align}

\paragraph{Three-term recurrence relation and difference equation.}
Using the definitions of $K_1$ and $\mathbf{a}_k$ given in \eqref{eq:DefK1realization} and \eqref{eq:eigVa}, respectively, one can show that the action of $K_1$ on the basis $\mathbf{a}_k$ is given by
\begin{align}
K_1\mathbf{a}_k&=-\frac{\left(k+2 a  -1\right) \left(N -k \right) \left(k-\rho +a  \right)}{2 \left(k+a  \right) \left(2k+2 a -1\right)}(\mathbf{a}_{k+1}-\mathbf{a}_k)\nonumber\\&-\frac12(N+\rho)\mathbf{a}_k
-\frac{k\left(k+a +\rho -1\right) \left(k+2 a +N  -1\right) }{2 \left(k +a -1\right) \left(2k+2 a  -1\right)}(\mathbf{a}_{k-1}-\mathbf{a}_k)\,.
\end{align}
Similarly, the action of $V$ on the basis $\mathbf{c}_\ell$ is given by 
\begin{align}
V\mathbf{c}_\ell&=(N-\ell)\mathbf{c}_{\ell+1}-\ell(\ell+a+\rho-1)(\ell+\rho-a-N)\mathbf{c}_{\ell-1}\nonumber\\
&+((\ell-N)(\ell+a+\rho)+\ell(\ell+\rho-a-N)-a(a-1))\mathbf{c}_\ell\,.
\end{align}
These formulas prove that $(V,K_1)$ forms a Leonard pair.

Combining these actions with \eqref{eq:overlap-ab-Hahn} and \eqref{eq:overlap-ba-Hahn}, one obtains the three-term recurrence relation and the difference equation for the Hahn polynomials \cite{Koekoek}. 

\section{Rational functions of Hahn type\label{sec:rational-functions}}
The Hahn rational functions 
\begin{align}
&U_k(\ell;a,c,N)={}_{3}F_2 \left({{-k,-\ell,k+2a-1}\atop
{ a-c-\ell,-N }}\;\Bigg\vert \; 1\right)\,\,,
\end{align}
are proportional to the overlap coefficients between the eigenbasis $\mathbf{a}_k$ and $\mathbf{b}_\ell$ defined in \eqref{eq:defn-en} and \eqref{eq:defn-bn}.
\begin{prop}
The following connection formulas hold
\begin{align}
&\mathbf{a}_k=\sum_{\ell=0}^N\binom{N}{\ell}(1-a-c-N)_{N-\ell}U_k(N-\ell,a,-c-N)\mathbf{b}_\ell\,,\label{eq:ak-sum-bl-rat-2}\\
&\mathbf{b}_\ell=(-1)^N\frac{(c-a)_\ell}{(2a)_N}\sum_{k=0}^N\frac{(1-2a-2k)}{(1-2a)}\frac{(-N,2a-1)_k}{k!(2a+N)_k} U_k(\ell;a,c-1,N)\mathbf{a}_k\label{eq:bl-sum-ak-rat}\,.
\end{align}
\end{prop}
\proof 
It is known that Pochhammer symbols are polynomials sequences of binomial type \cite{Rota}:
\begin{equation}
(x+y)_i=\sum_{\ell=0}^i\binom{i}{\ell}(x)_\ell(y)_{i-\ell}\,.
\end{equation}
Applying this identity, we obtain the connection formulas between the bases ${\mathbf{b}_\ell}$ and ${\mathbf{s}_i}$ in \eqref{eq:defn-bn} and \eqref{eq:basis-sn}, respectively
\begin{align}
&\mathbf{s}_i=\sum_{\ell=0}^i\binom{i}{\ell}(a-c)_{i-\ell}\mathbf{b}_\ell\,,\label{eq:overlap-sb}\\
&\mathbf{b}_\ell=\sum_{i=0}^\ell\binom{\ell}{i}(c-a)_{\ell-i}\mathbf{s}_i\,.
\end{align}
Additionally, we have given in the proof of Proposition \ref{prop:overlap-formulas-pols} the connection formula
\begin{align}
 \mathbf{a}_k&=\sum_{i=k}^N (-1)^{i+N}\binom{N-k}{i-k}(k+2a+i)_{N-i}\ \mathbf{s}_i\,.\label{eq:overlap-as}
\end{align}
Combining \eqref{eq:overlap-sb} and \eqref{eq:overlap-as}, we derive \eqref{eq:ak-sum-bl-rat-2}. The proof of \eqref{eq:bl-sum-ak-rat} follows the same steps.
\endproof

\paragraph{Biorthogonality} As a consequence of \eqref{eq:ak-sum-bl-rat-2} and \eqref{eq:bl-sum-ak-rat} we derive the biorthogonality of the Hahn rational functions: 
\begin{align}
    &\sum_{\ell=0}^N\mathcal{W}(\ell) U_s(\ell;a,c-1,N)U_k(N-\ell,a,-c-N)=h_k\delta_{k,s}\,,\\
    &\sum_{k=0}^N\frac{1}{h_k}U_k(\ell;a,c-1,N)U_k(N-\ell',a,-c-N)=\frac{\delta_{\ell,\ell'}}{\mathcal{W}(\ell)}
\end{align}
\begin{align}
&\mathcal{W}(\ell)=\binom{N}{\ell}(1-a-c-N)_{N-\ell}(c-a)_\ell\,,\\
&h_k=(-1)^N\frac{(1-2a)(2a)_N}{(1-2a-2k)}\frac{k!(2a+N)_k}{(-N,2a-1)_k}\,.
\end{align}

\subsection{Leonard trio\label{ssec:LT}}
The goal of this subsection is to indicate how the present work can be cast in the algebraic setting of Leonard trios introduced in \cite{LT2026}. Before doing so we recall their definition  (see \cite[Def. 2.2]{LT2026}). We denote by $\V$ a  vector space of finite dimension $N+1$ over the complex field $\mathbb{C}$, and $\mathrm{End}(\V)$ denotes the vector space over $\mathbb{C}$  of the endomorphisms of $\V$.
\begin{defi}\label{def:LT}
 A Leonard trio (LT) is an ordered triplet $(V,\oV,Z)$ of elements of $\mathrm{End}(\V)$ that satisfies the following properties:
   \begin{itemize}
       \item[(i)] There exists a basis of $\V$ with respect to which the matrix representing $V$ is diagonal and multiplicity free, and the matrices representing  $Z$ and $\oV Z$ are  tridiagonal.
       \item[(ii)] There exists a basis of $\V$ with respect to which the matrix representing $\oV$ is diagonal and multiplicity free, and the matrices representing  $Z$ and $Z V$ are  tridiagonal.
   \end{itemize}
\end{defi}
Next, we show that the triplet $(V,\oV,Z)$, with $V$, $\oV$, and $Z$ given in Paragraph \ref{par:realization-trio}, is a Leonard trio.\begin{prop}
The triplet $(V,\widetilde{V},Z)$ is a Leonard trio.
\end{prop}
\begin{proof}
As shown in \eqref{eq:eigVa}, the operator $V$ is diagonal in the basis $\mathbf{a}_k$. One checks that $Z$ and $\widetilde{V}Z$ are tridiagonal with respect to the basis ${\mathbf{a}_k}$:
\begin{align}
&Z\mathbf{a}_k=\frac{\left(N -k \right) \left(2 a +k -1\right)}{2 \left(k+a \right) \left(2 a +2 k -1\right)}
(\mathbf{a}_{k+1}-\mathbf{a}_k)-\mathbf{a}_k\nonumber\\
&\qquad -\frac{k \left(2 a +N +k -1\right)}{2 \left(k +a -1\right) \left(2 a +2 k -1\right)}
(\mathbf{a}_{k-1}-\mathbf{a}_k)\,,\label{eq:Z-ak}\\
&\widetilde{V}Z\mathbf{a}_k=-\frac{\left(N -k \right) \left(2 a +k -1\right)}{2(2 a +2 k -1)}(\mathbf{a}_{k+1}-\mathbf{a}_k)-(N+c)\mathbf{a}_k\nonumber\\
&\qquad \quad-\frac{k \left(2 a +N +k -1\right)}{2 \left(2 a +2 k -1\right)}(\mathbf{a}_{k-1}-\mathbf{a}_k)\label{eq:VtZ-ak}\,.
\end{align}
As shown in \eqref{eq:eigoVb}, the operator $\widetilde{V}$ is diagonal in the basis ${\mathbf{b}_\ell}$. One checks that $Z$ and $ZV$ are tridiagonal with respect to the basis $\mathbf{b}_\ell$:
\begin{align}
&Z\mathbf{b}_\ell=-\mathbf{b}_\ell+\ell \mathbf{b}_{\ell-1}\label{eq:Z-bl}\,,\\
&ZV\mathbf{b}_\ell=(N-l)\mathbf{b}_{\ell+1}-\ell(\ell+c-1)(\ell+c)\mathbf{b}_{\ell-1}\nonumber\\&\qquad\quad+(\ell(\ell+c)+(\ell-N)(\ell+c)+(a-1)(a+N))
\mathbf{b}_\ell\,.\label{eq:ZV-bl}
\end{align}
We conclude that the triplet $(V, \widetilde{V},Z)$ is a Leonard trio.
\end{proof}
Combining \eqref{eq:Z-ak}, \eqref{eq:VtZ-ak}, \eqref{eq:Z-bl}, and \eqref{eq:ZV-bl} with the connection formulas \eqref{eq:ak-sum-bl-rat-2} and \eqref{eq:bl-sum-ak-rat}, one obtains the generalized eigenvalue problems for the Hahn rational functions \cite{Tsujimoto_Vinet_Zhedanov_2025}.

\section{Outlook \label{sec:out}}
In this work we have completed the algebraic description of the Hahn type Leonard trios by explicitly determining the algebra generated by the operators $V$, $\widetilde V$, and $Z$ that define it. Complementing this set of operators with $Z^{-1}$ we could point relations that consistently define what we have called the trio Hahn algebra. 

A central result of this paper is the proof that the trio Hahn algebra is isomorphic to the previously introduced meta Hahn algebra. This establishes a precise algebraic equivalence between two constructions that are a priori different, thereby enriching the perspectives of the meta algebra setting and highlighting the role played by Leonard trios in organizing the bispectral properties of the associated special functions.

The realization in terms of difference operators further connects the abstract algebraic relations with explicit analytic models, allowing us to naturally recover the Hahn polynomials and biorthogonal rational functions.

The results of this paper support the view that together with meta algebras, Leonard trios provide an effective framework for encoding the bispectral properties of hypergeometric biorthogonal rational functons. They set the directions for the development of the program across other families within the meta algebra hierarchy. One natural direction for future work is to extend this analysis to other families, such as the $q$-Hahn \cite{bernard2024meta}, Racah \cite{MetaRacah,GM}, and Wilson \cite{wilson1991orthogonal}. We expect that this approach will lead to a systematic understanding the algebraic structures within the meta algebra hierarchy.

\vspace{1cm}
\paragraph{\textbf{Acknowledgements:}}
L.~Vinet is funded in part by a Discovery Grant from the Natural Sciences and Engineering Research Council (NSERC) of Canada. Q.~Labriet and L.~Morey enjoy postdoctoral fellowships provided by this grant. The work of L. Morey was supported by a mobility grant Québec-France Sophie Germain (https://doi.org/10.69777/381781) and thanks LAPTh for its hospitality.
\bibliographystyle{unsrt} 
%\bibliographystyle{unsrtinur} 
% formatting style : order of apparition numbering, number citations, 
% author initials, arXiv clickable URLs
\bibliography{ref_rosHahn.bib} 
\end{document}